\documentclass[aps,prl,showpacs,twocolumn,nofootinbib,10pt, superscriptaddress]{revtex4-1}
\pdfoutput=1
\usepackage{graphicx}
\usepackage{amsmath}
\usepackage{amssymb}
\usepackage{mathrsfs}
\usepackage{amsthm}
\usepackage{bm}
\usepackage{url}
\usepackage{csquotes}
\MakeOuterQuote{"}
\newtheoremstyle{note}
  {\topsep}               % ABOVE SPACE
  {\topsep}               % BELOW SPACE
  {}                      % BODY FONT
  {\parindent}            % INDENT (empty value is the same as 0pt)
  {\itshape}              % HEAD FONT
  {.}                     % HEAD PUNCTUATION
  {5pt plus 1pt minus 1pt}% HEAD SPACE
  {}

\theoremstyle{note}
\newtheorem{theorem}{Theorem}
\newtheorem{lemma}{Lemma}

\newtheorem{proposition}{Proposition}

\theoremstyle{definition}

\theoremstyle{remark}
\newtheorem{remark}{Remark}

\def\vec#1{\bm{#1}} %% overiding the original command

\providecommand{\tr}{\operatorname{tr}}
\providecommand{\Tr}{\operatorname{Tr}}

\newcommand{\rmT}{\mathrm{T}}

\newcommand{\be}{\begin{equation}}
\newcommand{\ee}{\end{equation}}
\newcommand{\ba}{\begin{align}}
\newcommand{\ea}{\end{align}}

\def\<{\langle}  %% overiding the original command \<
\def\>{\rangle}  %% overiding the original command \>

\newcommand{\dket}[1]{| #1\>\!\>}
\newcommand{\Dket}[1]{\Bigl| #1\Bigr\>\!\Bigr\>}
\newcommand{\dbra}[1]{\<\!\< #1|}
\newcommand{\Dbra}[1]{\Bigl\<\!\Bigl\< #1\Bigr|}

\newcommand{\dinner}[2]{\<\!\< #1| #2\>\!\>}

\newcommand{\douter}[2]{| #1\>\!\>\<\!\< #2|}

\newcommand{\barcal}[1]{\bar{\mathcal{#1}}}

\newcommand{\bid}{\bar{\mathbf{I}}}

\def\eqref#1{\textup{(\ref{#1})}}  %% overiding the original command \eqref
\newcommand{\eref}[1]{Eq.~\textup{(\ref{#1})}}

\newcommand{\esref}[1]{Eqs.~\textup{(\ref{#1})}}

\newcommand{\fref}[1]{Fig.~\ref{#1}}
\newcommand{\Fref}[1]{Figure~\ref{#1}}

\newcommand{\thref}[1]{Theorem~\ref{#1}}

\newcommand{\lref}[1]{Lemma~\ref{#1}}
\newcommand{\Lref}[1]{Lemma~\ref{#1}}

\newcommand{\pref}[1]{Proposition~\ref{#1}}

\newcommand{\psref}[1]{Propositions~\ref{#1}}

\newcommand{\cref}[1]{Conjecture~\ref{#1}}
\newcommand{\Cref}[1]{Conjecture~\ref{#1}}

\newcommand{\rcite}[1]{Ref.~\cite{#1}}
\newcommand{\rscite}[1]{Refs.~\cite{#1}}

\begin{document}

\title{Universal steering inequalities }
\author{Huangjun Zhu}
\email{Corresponding Author: hzhu@pitp.ca}
\affiliation{Perimeter Institute for Theoretical Physics, Waterloo, Ontario
N2L 2Y5, Canada}
\affiliation{Institute for Theoretical Physics, University of Cologne, Cologne
 50937, Germany}

\author{Masahito Hayashi}
\email{masahito@math.nagoya-u.ac.jp}
\affiliation{Graduate School of Mathematics, Nagoya University, Nagoya 464-0814,
Japan}
\affiliation{Centre for Quantum Technologies, National University of Singapore,
Singapore 117543, Singapore}

\author{Lin Chen}
\email{linchen@buaa.edu.cn}
\affiliation{School of Mathematics and Systems Science, Beihang University,
Beijing 100191, China}
\affiliation{International Research Institute for Multidisciplinary Science,
Beihang University, Beijing 100191, China}


%03.67.-a quantum information
%03.65.Wj quantum tomography, state reconstruction
%02.10.De algebraic structure
%03.65.-w quantum mechanics
%06.20.Dk: Measurement and error theory
%03.65.Ta: Foundations of quantum mechanics

%03.67.Mn Entanglement production, characterization and manipulation
%03.65.Ud Entanglement and quantum nonlocality
%(e.g. EPR paradox, Bell's inequalities, GHZ states, etc.)
%(for entanglement production in quantum information, see 03.67.Mn);

\begin{abstract}
We propose a general framework for constructing universal steering criteria  that are applicable to arbitrary bipartite states and measurement settings of the steering party. The same framework is also useful for studying the joint measurement problem. Based on the data-processing inequality for an extended R\'enyi relative entropy, 
we then introduce a family of universal steering inequalities, which detect steering much more efficiently than those inequalities known before.  
As illustrations, we show unbounded violation of a steering inequality for assemblages constructed from mutually unbiased bases and establish an interesting connection between maximally steerable assemblages and complete sets of  mutually unbiased bases. We also provide a single steering inequality that can detect all bipartite pure states of full Schmidt rank. 
In the course of study, we generalize a number of results intimately connected to data-processing inequalities, which are of independent interest.
\end{abstract}

\date{\today}
\maketitle

\emph{Steering} is a nonclassical phenomenon that formalizes what Einstein called "spooky action at a distance" \cite{EinsPR35, Schr35D}. For a long time, it was studied under the name of Einstein-Podolsky-Rosen (EPR) paradox \cite{Wern89, Reid89, ReidDBC09}. Recently, it was realized that steering is a form of nonlocality that sits between entanglement and Bell nonlocality \cite{WiseJD07,JoneWD07,SaunJWP10, QuinVCA15} and that is intrinsically asymmetric \cite{MidgFO10,BowlVQB14}.
Interestingly, steering  can be characterized by a simple quantum information processing task, namely, entanglement verification with an untrusted party \cite{WiseJD07,JoneWD07}.
 In addition, steering has been found useful in a number of applications, such as subchannel discrimination \cite{PianW15} and one-sided device-independent quantum key distribution \cite{BranCWS12}.

Recently, detection and characterization of steering have attracted increasing attention \cite{Reid89,WiseJD07,JoneWD07, ReidDBC09,CavaJWR09, WalbSGT11, SchnBWC13, Puse13, PramKM14, KogiSCA15, SkrzNC14,PianW15, KogiLRA15}. Various  steering criteria and inequalities have been derived, such as linear steering inequalities \cite{CavaJWR09, Puse13}; inequalities based on multiplicative variances  \cite{Reid89, ReidDBC09,CavaJWR09}, entropic uncertainty relations \cite{WalbSGT11, SchnBWC13},  fine-grained uncertainty relations \cite{PramKM14}; and hierarchy of steering criteria based on moments \cite{KogiSCA15}. However, most of these results are tailored to deal with specific scenarios;  majority criteria are only applicable to given numbers of measurement settings and outcomes. In addition, many criteria (including most linear  criteria) rely heavily on numerical optimization and lack a clear physical meaning and simple interpretation.

In this paper, we propose a general framework for constructing universal steering criteria that are applicable to arbitrary measurement settings of the steering party. In particular, we introduce nonlinear steering inequalities based on the data-processing inequality for an extended R\'enyi relative entropy \cite{Reny61,Haya06book}, which  detect  steering  more systematically and efficiently than  criteria in the literature. The same framework is also useful for studying the joint measurement problem \cite{BuscHL07, HeinW10, QuinVB14, UolaMG14, UolaBGP15, Zhu15IC}. In addition, our inequalities have a clear information theoretic meaning and simple interpretation.
As illustrations of the general framework, we show unbounded violation of a steering inequality by virtue of \emph{mutually unbiased bases} (MUB) \cite{WootF89, DurtEBZ10} and establish an interesting connection between maximally steerable settings and complete sets of MUB. We also provide a single steering inequality that can detect all bipartite pure states of full Schmidt rank.

\begin{figure}[bth]
\centering\noindent
\includegraphics[width=6cm]{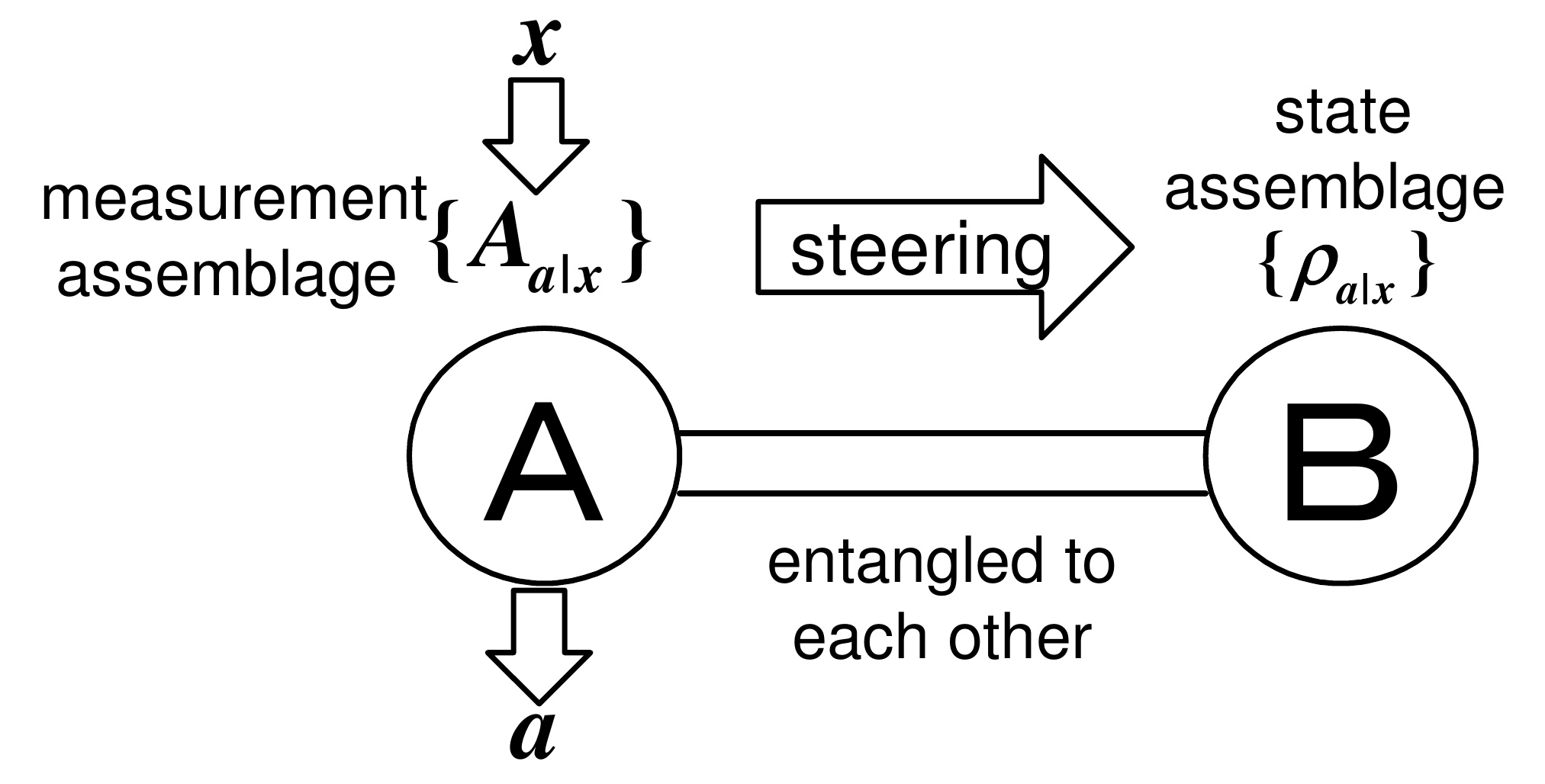}
\caption{\label{fig:steering}Steering scenario.
Alice can affect Bob's state via her choice of the measurement according to the relation  $\rho_{a|x}= \tr_A[(A_{a|x}\otimes 1) \rho]$. Entanglement is necessary but not sufficient for steering. }
\end{figure}

Suppose Alice and Bob share a bipartite state $\rho$ with reduced states $\rho_A$ and $\rho_B$.   Alice  can perform local measurements described by a collection of positive-operator-valued
measures (POVMs) $\{A_{a|x}\}$, which is known as a \emph{measurement assemblage}.  If Alice   obtains  outcome $a$ for  measurement $x$, then the unnormalized reduced state of Bob is $\rho_{a|x}=\tr_A [(A_{a|x}\otimes 1)\rho]$. In the following, we discuss  steering of Bob's system by Alice's measurements in terms of Bob's states~$\rho_{a|x}$.
The set of states $\rho_{a|x}$ for a given $x$ is called an \emph{ensemble} for $\rho_B$, and
the whose collection of ensembles is called a \emph{state
assemblage}~\cite{Puse13}; see \fref{fig:steering}.
To distinguish them, we express the ensemble by $\{\rho_{a|x}\}_a$ and the state assemblage by $\{\rho_{a|x}\}$. The assemblage $\{\rho_{a|x}\}$ is steerable if it does not admit a local hidden state model \cite{WiseJD07,JoneWD07} as $\rho_{a|x}=\sum_\lambda p(a|x,\lambda)\sigma_\lambda$ for all $a,x$,
where  $\{\sigma_\lambda\}$ is an ensemble for $\rho_B$ and   $p(a|x,\lambda)$ are a collection of stochastic maps with $p(a|x,\lambda)\geq 0$ and  $\sum_a p(a|x,\lambda)=1$. The state $\rho$ is steerable from Alice to Bob if there exists a measurement assemblage  
for Alice such that the resulting state assemblage for Bob is steerable. In this paper we shall focus on steerability of assemblages, no matter how they are constructed.

The steering problem is closely related to the joint measurement
problem of POVMs \cite{QuinVB14, UolaMG14, UolaBGP15, Zhu15IC, Puse15}.
Up to a scaling, a POVM may be seen as an ensemble for  the
completely mixed state. A measurement assemblage is \emph{compatible} or \emph{jointly measurable} if the corresponding state assemblage (for the completely mixed state) is unsteerable. 
In view of this connection, many results on steering can be turned into corresponding
results on POVMs, and vice versa. We shall make use of this connection without
further comments whenever convenient.

To set up the stage, we need to introduce  suitable order relations on ensembles and assemblages. 
Given two ensembles   $\{\rho_a\}$ and $\{\sigma_b\}$ for $\rho_B$, which may represent
two preparation procedures,  the ensemble  $\{\rho_a\}$  is a \emph{coarse graining} of $\{\sigma_b\}$, denoted by  $\{\rho_a\}\preceq \{\sigma_b\}$ or $\{\sigma_b\}\succeq \{\rho_a\}$,  if the former can be derived from the latter by data processing, that is, $\rho_a=\sum_b p(a|b) \sigma_b$,
where the stochastic map $p(a|b)$ characterizes the data-processing procedure. In that case, $\{\sigma_b\}$ is a \emph{refinement} of $\{\rho_a\}$. Intuitively, coarse graining usually leads to a less informative ensemble.
Two ensembles are \emph{equivalent} if they are coarse graining (refinement) of each other. The relation of coarse graining (refinement) forms a partial order on equivalent classes of ensembles for a given state.

The order relation on ensembles can be generalized to assemblages in a natural way. Given two assemblages $\{\rho_{a|x} \}$ and $\{\sigma_{b|y} \}$ for  $\rho_B$, the assemblage $\{\rho_{a|x} \}$ is a coarse graining of $\{\sigma_{b|y} \}$, denoted by $\{\rho_{a|x} \}\preceq \{\sigma_{b|y} \}$ or $\{\sigma_{b|y} \}\succeq\{\rho_{a|x} \}$, if each ensemble in $\{\rho_{a|x} \}$ is a coarse graining of an ensemble in  $\{\sigma_{b|y} \}$. In that case, $\{\sigma_{b|y} \}$ is called a refinement of $\{\rho_{a|x} \}$.
An assemblage is unsteerable if and only if it has a refinement that contains only one ensemble, that is,
all its ensembles possess a common refinement.
By definition, any coarse graining of an unsteerable assemblage is  unsteerable. Conversely, any refinement  of a steerable assemblage is steerable.

A function $f$ on ensembles is  \emph{order monotonic} (or order preserving) if $f(\{\rho_a\})\preceq f(\{\sigma_b\})$ whenever $\{\rho_a\}\preceq \{\sigma_b\}$. Order-monotonic functions on assemblages can be defined in a similar manner.
 Here the image of $f$ could be any space with a partial order, although we use the same notation for the partial order as that on ensembles.
 The image of all ensembles for a given state under an order-monotonic function $f$ is called the \emph{complementarity chamber}  and denoted by $\mathcal{C}_f$.
In those cases of interest to us, the chambers are usually   finite-dimensional compact convex sets,  and their shapes reflect the information tradeoff among different ensembles, hence the name.
For any unsteerable assemblage $\{\rho_{a|x}\}$ with a common refinement, say, $\{\sigma_\lambda\}$, we have $f(\{\rho_{a|x}\}_a)\preceq f(\{\sigma_\lambda\})\in  \mathcal{C}_f$ for all $x$. So $f(\{\rho_{a|x}\}_a)$
have a common upper bound in $\mathcal{C}_f$. Violation of this condition is a signature of steerability; see \fref{fig:USI} for an illustration.

\begin{figure}[tb]
\centering\noindent
\includegraphics[width=7cm]{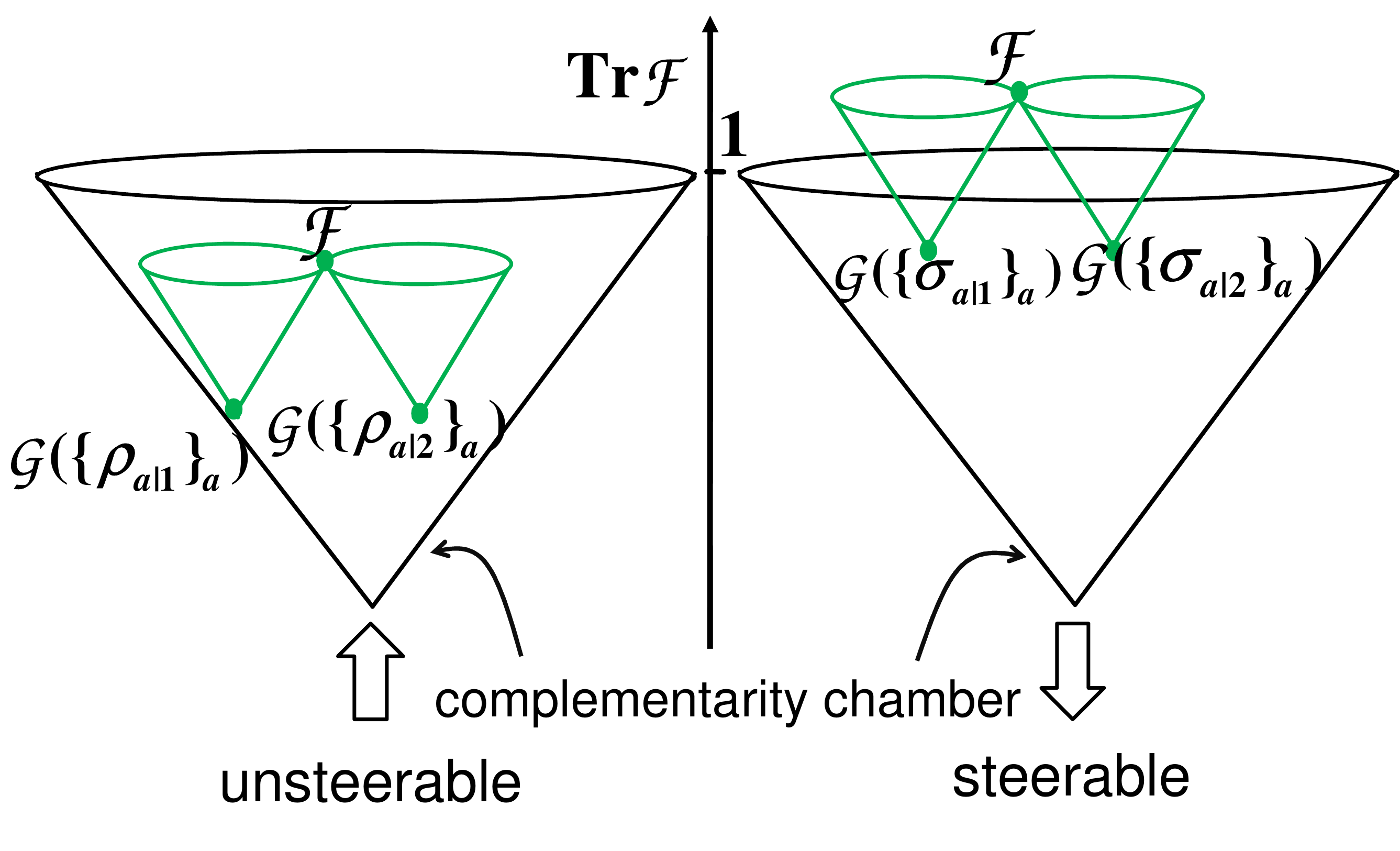}
\caption{\label{fig:USI}(color online) Simple idea behind universal steering criteria. Here the green cone at $\mathcal{G}(\{\rho_{a|x}\}_a)$ represents the set of superoperators $\mathcal{F}$ satisfying $\mathcal{F}\geq\mathcal{G}(\{\rho_{a|x}\}_a)$.
When  the  assemblage $\{\rho_{a|x}\}$ is unsteerable, $\mathcal{G}(\{\rho_{a|x}\}_a)$
have a common upper bound in the complementarity chamber (left plot). Violation
of this condition implies steerability (right plot).}
\end{figure}

To unleash the potential of the idea spelled out above, it is essential to construct  order-monotonic functions that are easy to characterize.
Inspired by the data-processing inequality for a R\'enyi relative entropy \cite{Reny61,Haya06book}, here 
we introduce two such functions from ensembles to superoperators.
Let $Q$ be a positive operator of full rank; define
\begin{equation}\label{eq:GQ}
\mathcal{G}_Q (\{\rho_a\}):=\sum_a \frac{\douter{\rho_a}{\rho_a} }{\tr(Q\rho_a)},\quad \barcal{G}_Q(\{\rho_a\}):=\sum_a \frac{\douter{\bar{\rho}_a}{\bar{\rho}_a} }{\tr(Q \rho_a)},
\end{equation}
where $\bar{\rho}_a=\rho_a-\tr(\rho_a)/d$ and $d$ is the dimension of the Hilbert space.  
Here, we consider the Hilbert space of operators on the physical space, i.e., the Hilbert-Schmidt space.
The kets in this space are denoted by the double-ket notation to distinguish them from ordinary kets. 
Superoperators, such as the outer product $\douter{A}{B}$, act on the operator space just as ordinary operators act on the usual Hilbert space; for example
$(\douter{A}{B})\dket{C}=\dket{A}\tr(B^\dag C)$ (cf.~\cite{Zhu12the,Zhu15IC}).

Now,  for a positive real vector $\vec{p}$ and a real vector $\vec{v}$ of the same length, 
we introduce extended R\'enyi relative entropy of order~2
as $D_2(\vec{v}\| \vec{p}):=\log \sum_k \frac{v_k^2}{p_k}$, which 
reduces to  the conventional R\'enyi relative entropy of order~2
when $\vec{p}$ and $\vec{v}$ represent probability distributions \cite{Reny61,Haya06book}.  
As shown in the supplementary material,
the  extended R\'enyi relative entropy obeys the data-processing inequality, from which we  deduce  the following theorem.
\begin{theorem}\label{thm:GQOM}
The functions $\mathcal{G}_Q(\cdot)$ and $\barcal{G}_Q(\cdot)$ are order monotonic  for any positive operator $Q$ of full rank.
\end{theorem}

When $Q$ is the identity, \eref{eq:GQ} reduces to
\begin{equation}\label{eq:G}
\mathcal{G}(\{\rho_a\}):=\sum_a \frac{\douter{\rho_a}{\rho_a} }{\tr(\rho_a)},\quad \barcal{G}(\{\rho_a\}):=\sum_a \frac{\douter{\bar{\rho}_a}{\bar{\rho}_a} }{\tr(\rho_a)}.
\end{equation}
We have
\begin{equation}\label{eq:Gbound}
\begin{aligned}
\Tr(\mathcal{G}(\{\rho_a\}))&=\sum_a\frac{\tr(\rho_a^2)}{\tr(\rho_a)}\leq \sum_a\tr(\rho_a)=1,\\
\Tr(\barcal{G}(\{\rho_a\}))&=\Tr(\mathcal{G}(\{\rho_a\}))-\frac{1}{d}\leq 1-\frac{1}{d},
\end{aligned}
\end{equation}
where "$\Tr$" denotes the trace of superoperators.
Here the upper bounds are saturated if and only if the ensemble is rank 1, that is, all the $\rho_a$ have rank 1. 
Define
\begin{equation}\label{eq:tau}
\begin{aligned}
\tau(\{\rho_{a|x}\})&=\min\{\Tr(\mathcal{F})| \mathcal{F}\geq \mathcal{G}(\{\rho_{a|x}\}_a)\; \forall x\},\\
\bar{\tau}(\{\rho_{a|x}\})&=\min\{\Tr(\mathcal{F})| \mathcal{F}\geq \barcal{G}(\{\rho_{a|x}\}_a)\; \forall x\}.
\end{aligned}
\end{equation}
\begin{theorem}\label{thm:USI}
The functions $\tau(\cdot)$ and $\bar{\tau}(\cdot)$ are order-monotonic on assemblages. Any unsteerable assemblage $\{\rho_{a|x}\}$ satisfies  $\tau(\{\rho_{a|x}\})\leq 1$ and $\bar{\tau}(\{\rho_{a|x}\})\leq 1-1/d$. Any compatible measurement assemblage $\{M_{a|x}\}$ satisfies $\tau(\{M_{a|x}\})\leq d$
 and $\bar{\tau}(\{M_{a|x}\})\leq d-1$.
\end{theorem}

\begin{proof}
Suppose $\{\rho_{a|x}\}\preceq \{\sigma_{b|y}\}$. Then for any ensemble $x$ in $\{\rho_{a|x}\}$ there exists an ensemble $y$ in
$\{\sigma_{b|y}\}$ such that $\mathcal{G}(\{\rho_{a|x}\}_a)\leq \mathcal{G}(\{\sigma_{b|y}\}_b)$ according to \thref{thm:GQOM}.
Therefore, $\mathcal{F} \geq \mathcal{G}(\{\rho_{a|x}\}_a)$ for all $x$ whenever $\mathcal{F} \geq\mathcal{G}(\{\sigma_{b|y}\}_b)$ for all $y$. It follows  that $\tau(\{\rho_{a|x}\})\leq \tau(\{\sigma_{b|y}\})$  and   $\tau(\cdot)$ is order-monotonic. By the same reasoning, so is $\bar{\tau}(\cdot)$.

The ensembles in $\{\rho_{a|x}\}$ possess a common refinement, say $\{\sigma_\lambda\}$, so that  $\mathcal{G}(\{\sigma_\lambda\})\geq \mathcal{G}(\{\rho_{a|x}\}_a)$ for all $x$. On the other hand, $\Tr(\mathcal{G}(\{\sigma_\lambda\}))\leq 1$ according to \eref{eq:Gbound}. It follows   that $\tau(\{\rho_{a|x}\})\leq 1$. The other three inequalities in \thref{thm:USI} follow from the same reasoning.
\end{proof}
\begin{remark}
According to \pref{pro:taudif} in
the supplementary material, $\tau(\{\rho_{a|x}\})=\bar{\tau}(\{\rho_{a|x}\})+1/d$, so  the  inequalities   $\tau(\{\rho_{a|x}\})\leq
1$ and $\bar{\tau}(\{\rho_{a|x}\})\leq 1-1/d$ are equivalent;  so are $\tau(\{M_{a|x}\})\leq d$ and $\bar{\tau}(\{M_{a|x}\})\leq d-1$. In practice, one  may be easier to analyze than another.

\end{remark}
\Fref{fig:USI} depicts the simple idea behind steering inequalities
in \thref{thm:USI}.
These inequalities means that unsteerable (compatible)\ assemblages cannot
be too informative (cf. \rcite{Zhu15IC}).
Compared with steering inequalities in the literatures \cite{CavaJWR09, Puse13},
what is remarkable is that they are applicable to arbitrary assemblages and
the bounds can be derived without numerical optimization.
The values of $\tau(\{\rho_{a|x}\})$ and $\bar{\tau}(\{\rho_{a|x}\})$ can
be computed efficiently with  semidefinite programming (SDP), whose size
increases only linearly with the number of
ensembles. Although the steerability of an assemblage can be determined by
 SDP \cite{Puse13, SkrzNC14}, the size of  such SDP increases exponentially
with  ensembles.  Our approach is attractive from both conceptual
and practical perspectives.

To illustrate the application of  our steering inequalities, we need to introduce several  concepts. Two ensembles $\{\rho_a\}$ and $\{\sigma_b\}$ are \emph{mutually orthogonal} if $\tr(\bar{\rho}_a\bar{\sigma}_b)=0 $ for all $a,b$ or, equivalently, if $\barcal{G}(\{\rho_{a}\})$ and  $\barcal{G}(\{\sigma_{b}\})$ have mutually orthogonal support.
The same definition applies to POVMs. For POVMs corresponding to rank-1 projective measurements, orthogonality is equivalent to mutually unbiasedness. Recall that two bases $\{|\psi_j\rangle \}$ and $\{|\varphi_k\rangle\}$ in dimension $d$ are mutually unbiased if $|\langle\psi_j|\varphi_k\rangle|^2=1/d$ for all $j,k$ \cite{WootF89, DurtEBZ10}. The following two propositions are proved in the supplementary material.
\begin{proposition}\label{pro:MSA}
Any measurement assemblage $\{M_{a|x}\}$ satisfies $\tau(\{M_{a|x}\})\leq d^2$ and $\bar{\tau}(\{M_{a|x}\})\leq d^2-1$. Any state  assemblage $\{\rho_{a|x}\}$   satisfies $\tau(\{\rho_{a|x}\})\leq d$ and $\bar{\tau}(\{\rho_{a|x}\})\leq d-1/d$.
\end{proposition}
\begin{proposition}\label{pro:MSAm}
Any measurement  assemblage $\{M_{a|x}\}$ with $m$ POVMs satisfies $\bar{\tau}(\{M_{a|x}\})\leq m(d-1)$. Any state assemblage $\{\rho_{a|x}\}$  with $m$ ensembles  satisfies $\bar{\tau}(\{\rho_{a|x}\})\leq m(1-1/d)$. The upper bound is saturated if and only if the POVMs (ensembles) are rank 1 and mutually orthogonal.
\end{proposition}

In view of \pref{pro:MSA}, measurement assemblages saturating the upper bound  $\bar{\tau}(\{M_{a|x}\})\leq d^2-1$ (or $\tau(\{M_{a|x}\})\leq
d^2$) are called maximally incompatible;  state  assemblages  saturating  $\bar{\tau}(\{\rho_{a|x}\})\leq d-1/d$ (or $\tau(\{\rho_{a|x}\})\leq d$) are called maximally steerable.

When  a measurement  assemblage $\{M_{a|x}\}$ is composed of $m$ projective measurements,  the upper bound  $\bar{\tau}(\{M_{a|x}\})\leq m(d-1)$ is saturated if and only if the bases are mutually unbiased. So measurement assemblages composed of MUB are maximally incompatible for given $m$; accordingly, state assemblages constructed from MUB are maximally steerable.
 The inequality $\bar{\tau}(\{M_{a|x}\})\leq d^2-1$ in \pref{pro:MSA} means that each set of MUB can contain at most $d+1$ bases, in agreement with the well-known bound \cite{WootF89, DurtEBZ10}.
When $d$ is a prime power, a complete set of MUB  can be constructed \cite{WootF89, DurtEBZ10}, so the compatibility inequality $\bar{\tau}(\{M_{a|x}\})\leq
d-1$ and the steering inequality $\bar{\tau}(\{\rho_{a|x}\})\leq 1-1/d$ can be violated by a factor of $d+1$, which is unbounded as $d$ grows. In contrast with the unbounded violation of a linear steering inequality shown in   \rcite{MarcRYH15}, our result follows from a universal recipe, and the degree of violation can be  determined precisely.
 An intriguing problem left open is how many bases are needed to construct a maximally incompatible (steerable) assemblage when complete sets of MUB cannot be found, say in dimension~6.

Next, we generalize \thref{thm:USI} by virtue of the order-monotonic functions $\mathcal{G}_Q$ and $\barcal{G}_Q$.  Define   superoperator $\mathcal{R}_Q$~\cite{BrauC94,PetzS96,Zhu12the,Zhu15IC} and $\barcal{R}_Q$ by the
equation
\begin{equation}\label{eq:LRmultiplication}
\begin{aligned}
\mathcal{R}_Q\dket{S}=\frac{1}{2}\dket{SQ+ QS},\quad
\barcal{R}_Q=\mathcal{R}_Q-\frac{\douter{Q}{Q}}{\tr(Q)}.
\end{aligned}
\end{equation}
Note that both  $\mathcal{R}_Q$ and $\barcal{R}_Q$ are positive superoperators, that is, $\dbra{S}\mathcal{R}_Q\dket{S}\geq 0$ and $\dbra{S}\barcal{R}_Q\dket{S}\geq 0$ for any operator $S$; so $\mathcal{R}_Q^{1/2}$ and $\barcal{R}_Q^{1/2}$ are well defined. Here comes the analogy of \eref{eq:Gbound},
\begin{equation}\label{eq:GQbound}
\begin{aligned}
\Tr(\mathcal{R}_Q \mathcal{G}_Q(\{\rho_a\}))&=\sum_a\frac{\tr (Q\rho_a^2)}{\tr(Q\rho_a)}\leq \sum_a\tr(\rho_a)=1,\\
\Tr(\barcal{R}_Q \barcal{G}_Q(\{\rho_a\}))&=\Tr(\mathcal{R}_Q \mathcal{G}_Q(\{\rho_a\}))-\frac{\tr(Q\rho_B)}{\tr (Q)} \\
&\leq 1-\frac{\tr(Q\rho_B)}{\tr(Q)} .
\end{aligned}
\end{equation}
Again,  the upper bounds are saturated if and only if the ensemble is rank 1. The operator $Q$ serves as a probe. If $Q=1$, then $\mathcal{R}_Q=\mathbf{I}$ and $\barcal{R}_Q=\bid$, where $\mathbf{I}$ is the identity superoperator and $\bid$ is the projector onto the space of traceless operators, so \eref{eq:GQbound} reduces to \eref{eq:Gbound}. Define
\begin{equation}\label{eq:tauQ}
\begin{aligned}
\tau_Q(\{\rho_{a|x}\})&=\min\{\Tr(\mathcal{F})| \mathcal{F}\geq \tilde{\mathcal{G}}_Q(\{\rho_{a|x}\}_a)\; \forall x\},\\
\bar{\tau}_Q(\{\rho_{a|x}\})&=\min\{\Tr(\mathcal{F})| \mathcal{F}\geq \tilde{\barcal{G}}_Q(\{\rho_{a|x}\}_a)\; \forall x\},
\end{aligned}
\end{equation}
where $\tilde{\mathcal{G}}_Q=\mathcal{R}_Q^{1/2}\mathcal{G}_Q\mathcal{R}_Q^{1/2}$ and $\tilde{\barcal{G}}_Q=\barcal{R}_Q^{1/2}\barcal{G}_Q\barcal{R}_Q^{1/2}$. By the same reasoning as in the proof of \thref{thm:USI}, we have
\begin{theorem}\label{thm:USI2}
The functions $\tau_Q(\cdot)$ and $\bar{\tau}_Q(\cdot)$ are order-monotonic on assemblages for any invertible positive operator $Q$. Any unsteerable state assemblage $\{\rho_{a|x}\}$ satisfies
$\tau_Q(\{\rho_{a|x}\})\leq 1$ and $\bar{\tau}_Q(\{\rho_{a|x}\})\leq 1-\tr(Q\rho_B)/\tr(Q)$.
Any compatible measurement assemblage $\{\rho_{a|x}\}$ satisfies
$\tau_Q(\{M_{a|x}\})\leq d$ and $\bar{\tau}_Q(\{M_{a|x}\})\leq d-1$.
\end{theorem}

To illustrate the power of \thref{thm:USI2}, let us consider a measurement assemblage $\{M_{a|x}\}$ composed of two different symmetric informationally complete measurements (SICs) \cite{Zaun11, ReneBSC04}. Since SICs form 2-designs and tight informationally complete measurements \cite{Scot06, ZhuE11, Zhu12the,ApplFZ15G}, we have 
$\mathcal{G}(\{M_{a|1}\}_a)=\mathcal{G}(\{M_{a|2}\}_a)=(\mathbf{I}+\douter{1}{1})/(d+1)$ and $\tau(\{M_{a|x}\})= d$. If $Q$ is a generic quantum state, then  $\mathcal{G}_Q(\{M_{a|1}\}_a)\neq\mathcal{G}_Q(\{M_{a|2}\}_a)$, so that $\tau_Q(\{M_{a|x}\})>d$.
 In this case, the inequalities in \thref{thm:USI2}  can detect incompatibility (steering) that cannot be detected by  \thref{thm:USI} thanks to the choice in  the probe $Q$.

More general steering inequalities can be derived by  considering the effect of filtering. The following proposition is an easy generalization of a result in \rcite{UolaBGP15}.
\begin{proposition}\label{pro:FilterSteer}
The two assemblages $\{V\rho_{a|x}V^\dag\}$ and  $\{V\rho_{a|x}^\rmT V^\dag\}$ (unnormalized) are both unsteerable for any operator $V$ if $\{\rho_{a|x}\}$ is unsteerable.
When $V$ is invertible, $\{V\rho_{a|x}V^\dag\}$,  $\{V\rho_{a|x}^\rmT V^\dag\}$, and $\{\rho_{a|x}\}$ are simultaneously steerable or not.
\end{proposition}
When Bob's state $\rho_B$ is invertible, \thref{thm:USI2} and 
\pref{pro:FilterSteer} imply that 
any unsteerable assemblage $\{\rho_{a|x}\}$ satisfies  
\begin{equation}\label{eq:SteeringIneqGQrhoB}
\tau_Q(\{\rho_B^{-\frac{1}{2}} \rho_{a|x}\rho_B^{-\frac{1}{2}}\})\leq d,
\quad
 \bar{\tau}_Q(\{\rho_B^{-\frac{1}{2}}\rho_{a|x}\rho_B^{-\frac{1}{2}}\})\leq
d-1.
\end{equation}

Until now, we  discuss steering  in terms of Bob's state assemblage $\{\rho_{a|x}\}$.
At this point, it is instructive to consider steering of Bob's state by Alice's measurements as described by the assemblage 
$\{A_{a|x}\}$, which is the physical situation illustrated in \fref{fig:steering}.
Suppose they share a pure bipartite state $\rho$ of full Schmidt rank, which has the Schmidt decomposition $\rho=\sum_{j,k}\lambda_j\lambda_k |jj\rangle\langle kk|$. Then the reduced states of Alice and Bob have the same form with respect to the Schmidt basis 
$\rho_{A}=\rho_{B}=\sum_j \lambda_j^2 |j\rangle\langle j|$, and the state assemblage $\{\rho_{a|x}\}$ for Bob takes on the form
$\rho_{a|x}=\rho_B^{1/2}A_{a|x}^\rmT\rho_B^{1/2}$ \cite{Haya06book,UolaBGP15}. Therefore,
\begin{equation}\label{eq:H}
\begin{aligned}
{\tau}_Q(\{\rho_B^{-1/2} \rho_{a|x}\rho_B^{-1/2}\})
={\tau}_Q(\{A_{a|x}^\rmT\}),
\\
\bar{\tau}_Q(\{\rho_B^{-1/2} \rho_{a|x}\rho_B^{-1/2}\})=\bar{\tau}_Q(\{A_{a|x}^\rmT\}).
 \end{aligned}
\end{equation}
As a consequence of  \esref{eq:SteeringIneqGQrhoB} and \eqref{eq:H},
$\tau(\{A_{a|x}\}) \le d$ and $\bar{\tau}(\{A_{a|x}\}) \le d-1$ if  Alice cannot steer Bob's system, note that ${\tau}(\{A_{a|x}^\rmT\}) 
={\tau}(\{A_{a|x}\})$ and $\bar{\tau}(\{A_{a|x}^\rmT\})
=\bar{\tau}(\{A_{a|x}\})$ according to \lref{lem:tauTranspose}
in the supplementary
material.
If  the   assemblage $\{A_{a|x}\}$ is composed of two MUB, then $\bar{\tau}(\{
A_{a|x}\})=2(d-1)$, 
which violates the second inequality by a factor of 2. 
Remarkably, the single steering inequality 
$\bar{\tau}(\{\rho_B^{-\frac{1}{2}}\rho_{a|x}\rho_B^{-\frac{1}{2}}\}) \le d-1$ with two measurement settings can detect
the steerability of all bipartite pure states of full Schmidt rank, 
whereas 
infinitely many inequalities linear in $\rho_{a|x}$ 
(note that our inequalities are not linear in $\rho_{a|x}$) 
are needed to achieve the same task \cite{CavaJWR09, Puse13}.  
Also,  no general recipe is known for constructing linear steering
inequalities without numerical optimization.
Therefore, our approach provides a dramatic improvement over those
alternatives known in the literatures. Further, an additional example on isotropic states is presented in the supplementary material.

In summary, we proposed a general framework for detecting and characterizing
steering based on simple information theoretic ideas.
Based on the data-processing inequality for the extended R\'enyi relative entropy of
order~2, we then introduced a family of universal steering inequalities that are applicable
to arbitrary assemblages and  have a simple interpretation. 
As illustrations,
we showed unbounded violation of a steering inequality for assemblages constructed
from MUB and provided a single steering inequality that can detect all bipartite
pure states of full Schmidt rank. 
Our work established intriguing  connections
among a number of fascinating subjects, including  information
theory, quantum foundations, and geometry of quantum state space, which are
of interest to researchers from diverse fields.
In addition, our work has an intimate connection to quantum estimation theory.
Indeed, our \thref{thm:GQOM} can be applied to prove  the  data-processing inequality
for Fisher information, and vice versa (cf. the supplementary material). Also, our study allows to derive and generalize many results in quantum estimation theory, which will be presented in another paper.

\section*{Acknowledgments}
HZ is grateful to Matthew Pusey and Joshua Combes for comments and discussions. This work
is supported  by Perimeter Institute for Theoretical Physics. Research at
Perimeter Institute is supported by the Government of Canada through Industry
Canada and by the Province of Ontario through the Ministry of Research and
Innovation. HZ also acknowledges financial support from the Excellence Initiative
of the German Federal and State Governments
(ZUK 81) and the DFG. MH is partially supported by a MEXT Grant-in-Aid for
Scientific Research (A) No. 23246071 and the National Institute of Information
and Communication Technology (NICT), Japan.
The Centre for Quantum Technologies is funded by the
Singapore Ministry of Education and the National Research Foundation
as part of the Research Centres of Excellence programme.
LC was supported by the NSF of China (Grant No. 11501024), and the Fundamental
Research Funds for the Central Universities (Grant Nos. 30426401 and 30458601).

\bigskip
\bibliographystyle{apsrev4-1}
\bibliography{all_references}

 \clearpage
%%%%%%%%% Merge with supplemental materials %%%%%%%%%%
%%%%%%%%% Prefix a "S" to all equations, figures, tables and reset the counter
%%%%%%%%%
\setcounter{equation}{0}
\setcounter{figure}{0}
\setcounter{table}{0}
\setcounter{theorem}{0}
\setcounter{lemma}{0}
\setcounter{remark}{0}

\makeatletter
\renewcommand{\theequation}{S\arabic{equation}}
\renewcommand{\thefigure}{S\arabic{figure}}
\renewcommand{\thetable}{S\arabic{table}}
\renewcommand{\thetheorem}{S\arabic{theorem}}
\renewcommand{\thelemma}{S\arabic{lemma}}
\renewcommand{\theremark}{S\arabic{remark}}

%\renewcommand{\bibnumfmt}[1]{[S#1]}
%\renewcommand{\citenumfont}[1]{S#1}
%%%%%%%%%% Prefix a "S" to all equations, figures, tables and reset the counter
%%%%%%%%%%

 \onecolumngrid
 \begin{center}
 \textbf{\large Supplementary material: Universal steering inequalities}
 \end{center}
In this supplementary material
we prove \thref{thm:GQOM} in the main text
 by proving the data-processing inequality for the extended R\'enyi relative entropy of order~2.
We also explain the connection between \thref{thm:GQOM} and  generalized data-processing inequalities for generalized relative entropies as well as  the data-processing inequality for 
Fisher information.  
Quite surprisingly, \thref{thm:GQOM} embodies a generalization of two distinct
data-processing inequalities. This observation reveals an intriguing connection between the R\'enyi relative entropy and Fisher information, which deserves further study.
In addition, we derive an equality between $\tau(\{\rho_{a|x}\})$ and $\bar{\tau}(\{\rho_{a|x}\})$, and prove \psref{pro:MSA} and \ref{pro:MSAm} in the main text,
which are needed for determining maximally steerable assemblages. We also
show  that the  functions $\tau(\{\rho_{a|x}\})$ and $\bar{\tau}(\{\rho_{a|x}\})$
introduced in the main text are invariant under transposition of $\rho_{a|x}$.
Finally, we  consider
the steerability of isotropic states as another illustration of  our approach.

\section{Proof of \thref{thm:GQOM}}
\begin{proof}
Let $A=\{A_k\}$ and $B=\{B_j\}$ be two sets of positive operators that satisfy
 $B_j=\sum_k\Lambda_{jk}A_k$ for an arbitrary stochastic matrix~$\Lambda$.
Let $C$ be an arbitrary Hermitian operator; then
\begin{equation}
\dbra{C}\mathcal{G}_Q(A)\dket{C}=\sum_k \frac{v_k^2}{p_k}, \quad \dbra{C}\mathcal{G}_Q(B)
\dket{C}=\sum_j \frac{u_j^2}{q_j},
\end{equation}
where  $p_k=\tr(Q A_k)$, $q_j=\tr(Q
B_j)$, $v_k=\tr(A_k C)$, and $u_j=\tr(B_j C)$ satisfy  $q_j=\sum_k\Lambda_{jk} p_k$ and $u_j=\sum_k\Lambda_{jk}  v_k$.
According to \lref{lem:DPIRenyi2} below, $\dbra{C}\mathcal{G}_Q(B)\dket{C}\leq
\dbra{C}\mathcal{G}_Q(A)\dket{C}$, which implies that $\mathcal{G}_Q(B)\leq
\mathcal{G}_Q(A)$ and $\mathcal{G}_Q(\cdot)$ is order monotonic. By the same
token, so is $\barcal{G}_Q(\cdot)$.
Alternatively, the latter conclusion follows from the observation that $\barcal{G}_Q(\cdot)=\bid\mathcal{G}_Q(\cdot)\bid$.\end{proof}

\begin{remark}
$\{A_k\}$ and $\{B_j\}$ in the above proof are not necessarily normalized
assemblages as long as they are connected by a stochastic matrix. Therefore,
\thref{thm:GQOM} is applicable for both normalized and unnormalized assemblages.
\end{remark}

Let $\vec{p}$ and $\vec{v}$ be two real vectors of the same length and
$p_k>0$ for all $k$. Following the main text, the extended R\'enyi relative entropy of order~2 between the two vectors is defined as
$D_2(\vec{v}\| \vec{p}):=\log \sum_k \frac{v_k^2}{p_k}$.

\begin{lemma}\label{lem:DPIRenyi2}
Given  two real vectors $\vec{p}$ and $\vec{v}$ as above,
let $\vec{q}=\Lambda
\vec{p}$ and $\vec{u}=\Lambda \vec{v}$, where $\Lambda$ is a stochastic matrix.
Then $D_2(\vec{u}\| \vec{q})\leq D_2(\vec{v}\| \vec{p})$, that is,
\begin{equation}\label{eq:DPIRenyi2}
\sum_j \frac{u_j^2}{q_j}\leq \sum_k \frac{v_k^2}{p_k}.
\end{equation}
\end{lemma}
\begin{proof}
The proof follows the same idea as in the proof of data-processing
inequalities for generalized relative entropies~\cite{Csis67,Haya06book}. It relies
on the convexity of the quadratic function.
\begin{align}
\sum_j \frac{u_j^2}{q_j}&=\sum_j q_j\biggl(\sum_k\frac{\Lambda_{jk}p_k}{q_j}
\frac{v_k}{p_k}\biggr)^2\leq  \sum_j q_j\sum_k\frac{\Lambda_{jk}p_k}{q_j}\Bigl(
\frac{v_k}{p_k}\Bigr)^2=\sum_k\frac{v_k^2}{p_k},
\end{align}
note that $\{\frac{\Lambda_{jk}p_k}{q_j}\}_k$ forms a probability distribution.
\end{proof}
\begin{remark}
Here we assume that no row of $\Lambda$ are identically 0, so that $q_j>0$ for all $j$.
\end{remark}

When $\vec{v}$ and $\vec{p}$ are probability distributions, 
$D_2(\vec{v}\| \vec{p})$ is 
R\'enyi relative entropy of order~2, which obeys a well-known data-processing inequality \cite{Reny61,Haya06book}. In \lref{lem:DPIRenyi2}, the four vectors
$\vec{q},\vec{p},  \vec{u},\vec{v}$
are not necessarily probabilities as long as the components of
$\vec{p}$ and $\vec{q}$ are positive. 
Therefore, \eref{eq:DPIRenyi2} may be understood as a generalization of the data-processing inequality for R\'enyi relative entropy of order~2.
Although this relative entropy and its variant
play important roles in information theory and cryptography theory \cite{BennBCM95,HastILL99}, 
we are not aware of 
the extension of the quantity 
$\log \sum_k \frac{v_k^2}{p_k}$ and \eref{eq:DPIRenyi2}
for general vectors in the literature.
Our work may stimulate further progress along this direction. 

\section{Generalized  data-processing inequalities}
In this section, we explore potential extension of our approach based on generalized data-processing inequalities. 
Let $f$ be a convex function defined on  real numbers; 
let $\vec{p}$ and $\vec{v}$ be two real vectors of the same length such that $p_k>0$ for all $k$. Define
the extended $f$-relative entropy as
\begin{equation}
D_f(\vec{v}\| \vec{p})=\sum_k p_k f\Bigl(\frac{v_k}{p_k}\Bigr).
\end{equation}

\begin{theorem}
$D_f(\vec{v}\|\vec{p})$ is monotonic under data processing for any convex
function $f$, that is, 
\begin{equation}\label{eq:DPIRE}
D_f(\Lambda\vec{v}\| \Lambda\vec{p})\leq D_f(\vec{v}\|\vec{p})
\end{equation}
for any stochastic matrix $\Lambda$. 
\end{theorem}
\begin{remark}
This theorem can be proved using Jensen's inequality for a convex function
according to the same reasoning as in the proof of \lref{lem:DPIRenyi2} 
(cf. \rscite{Csis67,Haya06book}). 
\end{remark} 

When $\vec{p}$ and $\vec{v}$ are probability distributions, $D_f(\vec{v}\|\vec{p})$ 
is known as the  $f$-relative entropy  ($f$-divergence) between  $\vec{p}$ and $\vec{v}$.
It quantifies  the closeness between the two probability distributions, which
is intimated connected to mutual information. In that case, \eref{eq:DPIRE}
reduces to the data-processing inequality (also known as information-processing
inequality) for $f$-relative entropy; cf. \rcite{Csis67, Haya06book}.
In addition, we can relax the  requirement on $f$ to 
be a convex function defined over positive numbers.
When $f(x)=x\log x$, $D_f(\vec{v}\| \vec{p})$ reduces to the usual relative
entropy, and \eref{eq:DPIRE} represents the well known data-processing inequality
for the relative entropy. 
When $f(x)=x^2$, $\log D_f(\vec{v} \| \vec{p})$ coincides with the R\'enyi relative entropy of order~2, and \eref{eq:DPIRE}
takes on the same form as \eref{eq:DPIRenyi2}.

In our application, $v_k$ may take on negative values (this is the case for  $v_k=\tr(A_k C)$ in the proof of \thref{thm:GQOM}), so it is desirable to choose the function $f$ that can accept a negative entry such as $f(x)=x^2$. This is the reason we introduce the extended $f$-relative entropy.
Other potential choices include the quartic function $f(x)=x^4$.
We leave it open whether such functions can be used to construct steering inequalities.

\section{Connection with the data-processing inequality for
Fisher information}
In this section,
we explain the connection between the order-monotonic  functions $\mathcal{G}_Q(\cdot)$ and $\barcal{G}_Q(\cdot)$ in Theorem 1 and 
 Fisher information with regard to the data-processing inequality \cite{Zami98}. 
Consider a random variable $A$ as the outcome of a measurement or observation.
Suppose $A$ takes on  values over positive integers and the probability
of obtaining outcome $j$ is $p_A(j|\theta)$, where $\theta$ is the unknown
parameter. 
The \emph{Fisher information}  associated with $A$ is given by
\begin{align}\label{sym:FisherInf}
I_A(\theta)&= \sum_j
p_A(j|\theta)\Bigl(\frac{\partial \ln p_A(j|\theta)}{\partial
\theta}\Bigr)^2=\sum_j\frac{1}{p_A(j|\theta)}\Bigl(\frac{\partial
p_A(j|\theta)}{\partial \theta}\Bigr)^2.
\end{align}

\begin{lemma}\label{lem:chainRule}
Suppose $A, B$ are two random variables with joint distribution $p(j,k|\theta)$.
Then
the total Fisher information provided by the two random variables is given by\begin{equation}
I_{A,B}(\theta)=I_A(\theta)+ I_{B|A}(\theta),
\end{equation}
where
\begin{equation}
I_{B|A}(\theta)=\sum_{j,k}p(j,k|\theta)  \biggl(\frac{\partial \ln p_{B|A}(k|j;\theta)}{\partial
\theta}\biggr)^2
\end{equation}
is the conditional Fisher information, and $p_{B|A}(k|j;\theta)$ is the conditional probability distribution.
\end{lemma}
\begin{remark}
This chain rule for Fisher information and the data-processing inequality
in \thref{thm:FIDPI} below were derived
in \rcite{Zami98} (in a slightly different form). \end{remark}

\begin{proof}
\begin{align}
&I_{A,B}(\theta)=\sum_{j,k}
p(j,k|\theta)\Bigl(\frac{\partial \ln p(j,k|\theta)}{\partial
\theta}\Bigr)^2=\sum_{j,k} p(j,k|\theta) \biggl(\frac{\partial \ln p_A(j|\theta)}{\partial
\theta}+\frac{\partial \ln p_{B|A}(k|j;\theta)}{\partial \theta}\biggr)^2\nonumber\\
&=\sum_{j,k} p(j,k|\theta) \biggl[\biggl(\frac{\partial \ln p_A(j|\theta)}{\partial
\theta}\biggr)^2+ \biggl(\frac{\partial \ln p_{B|A}(k|j;\theta)}{\partial
\theta}\biggr)^2\biggr]+  2\sum_{j,k} p(j,k|\theta)\frac{\partial \ln p_A(j|\theta)}{\partial
\theta}\frac{\partial \ln p_{B|A}(k|j;\theta)}{\partial \theta} \nonumber
\\
&=I_A(\theta) +I_{B|A}(\theta),
\end{align}
where in deriving the last equality,  we have employed the following equation
\begin{align}
&\sum_{j,k} p(j,k|\theta)\frac{\partial \ln  p_A(j|\theta)}{\partial \theta}\frac{\partial
\ln p_{B|A}(k|j;\theta)}{\partial \theta}\nonumber \\
&=\sum_j\biggl( p_A(j|\theta)\frac{\partial \ln  p_A(j|\theta)}{\partial
\theta} \sum_k p_{B|A}(k|j;\theta)\frac{\partial \ln p_{B|A}(k|j;\theta)}{\partial
\theta}\biggr)\nonumber \\
&=\sum_j \biggl( p_A(j|\theta)\frac{\partial \ln  p_A(j|\theta)}{\partial
\theta}\sum_k\frac{\partial  p_{B|A}(k|j;\theta)}{\partial \theta}\biggr)=0,
\end{align}
note that $\sum_kp_{B|A}(k|j;\theta)=1$.
\end{proof}

\begin{theorem}\label{thm:FIDPI}
Suppose $A, B$ are two random variables whose marginal  distributions  satisfy
$p_B(j|\theta)=\sum_k\Lambda_{jk}p_A(k|\theta)$, where $\Lambda$ is a stochastic
matrix  independent of the parameter $\theta$.
Then $I_B(\theta)\leq I_A(\theta)$.
\end{theorem}
This theorem is an immediate consequence of \lref{lem:chainRule}, given  that
$I_{B|A}(\theta)=0$ since the conditional probability distribution $p_{B|A}$
is independent of $\theta$.

In the rest of this section we show that the superoperators $\mathcal{G}_Q(A)$
and  $\barcal{G}_Q(A)$ introduced in the main text may be understood as Fisher
information matrices in superoperator form when $Q$ is a quantum state and
$A=\{A_k\}$ is a POVM.   Note that  $\tr(Q A_k)$ is
the probability of obtaining outcome $A_k$
upon measuring $A$ on $Q$  and that
\begin{equation}
\mathcal{G}_Q(A)=\sum_k \frac{\douter{A_k}{A_k} }{\tr(Q A_k)},\quad \barcal{G}_Q(A)=\sum_k
\frac{\douter{\bar{A}_k}{\bar{A}_k} }{\tr(Q A_k)}.
\end{equation}
Suppose $Q$ is a
quantum state that depends on the parameter $\theta$. Then the Fisher information
concerning $\theta$ provided
by the measurement $A$ is 
\begin{equation}
I_A(\theta)=\sum_k
\frac{1}{\tr(Q A_k)}\tr\Bigl(\frac{\partial
Q}{\partial \theta}A_k\Bigr)^2=\Dbra{\frac{\partial
Q}{\partial \theta}}\mathcal{G}_Q(A)\Dket{\frac{\partial
Q}{\partial \theta}}=\Dbra{\frac{\partial
Q}{\partial \theta}}\barcal{G}_Q(A)\Dket{\frac{\partial
Q}{\partial \theta}},
\end{equation}
where the third equality follows from the fact that $\partial
Q/\partial \theta$ is traceless. Therefore, $\mathcal{G}_Q(A)$ and  $\barcal{G}_Q(A)$
are Fisher information matrices in disguise, and  \thref{thm:GQOM} in the
main text embodies generalized data-processing inequalities for Fisher information.
As an implication of our result, the Fisher information data-processing inequality
is also valid for incomplete observations or measurements.  These results are  of interest beyond the focus of this paper,
such as in quantum metrology \cite{CombFJC14,Ferr14}.

\section{Connection between $\tau(\{\rho_{a|x}\})$ and $\bar{\tau}(\{\rho_{a|x}\})$}

\begin{proposition}\label{pro:taudif}Any state assemblage $\{\rho_{a|x}\}$ satisfies
$\tau(\{\rho_{a|x}\})=\bar{\tau}(\{\rho_{a|x}\})+1/d$. Any measurement assemblage $\{M_{a|x}\}$ satisfies  $\tau(\{M_{a|x}\})=\bar{\tau}(\{M_{a|x}\})+1$.
\end{proposition}
\begin{proof}Let $\Delta_x=\mathcal{G}(\{\rho_{a|x}\}_a)-\barcal{G}(\{\rho_{a|x}\}_a)$. Calculation shows that $\Delta_x$ is independent of $x$, 
\begin{equation}
\Delta_x=\Delta:=\frac{\douter{1}{\rho_B}+\douter{\rho_B}{1}}{d}-\frac{\tr(\rho_B)}{d^2}(\douter{1}{1}),\quad \Tr(\Delta_x)=\Tr(\Delta)=\frac{\tr(\rho_B)}{d}=\frac{1}{d}.
\end{equation}
If $\mathcal{F}\geq \mathcal{G}(\{\rho_{a|x}\}_a)$ for all $x$, then $\mathcal{F}-\Delta\geq \barcal{G}(\{\rho_{a|x}\}_a)$ for all $x$. So $\bar{\tau}(\{\rho_{a|x}\})\leq \tau(\{\rho_{a|x}\})-\Tr(\Delta)=\tau(\{\rho_{a|x}\})-1/d$. Similarly, we have   the inequality $\tau(\{\rho_{a|x}\})\leq \bar{\tau}(\{\rho_{a|x}\})+1/d$. It follows that $\tau(\{\rho_{a|x}\})=\bar{\tau}(\{\rho_{a|x}\})+1/d$. The equality $\tau(\{M_{a|x}\})=\bar{\tau}(\{M_{a|x}\})+1$ follows from the same reasoning.
\end{proof}

\section{Proofs of \psref{pro:MSA} and \ref{pro:MSAm}}
In this section we prove \psref{pro:MSA} and \ref{pro:MSAm} in the main
text, which are needed for determining maximally steerable assemblages.

\begin{proof}[Proof of \pref{pro:MSA}]
According to \lref{lem:SLDbound} below, $\tau(\{M_{a|x}\})\leq \Tr(\mathbf{I})
= d^2$ and $\bar{\tau}(\{M_{a|x}\})\leq \Tr(\bid) =d^2-1$. 
Here $\mathbf{I}$ is the identity superoperator and $\bid$ is the projector
onto
the space of traceless operators. By the same token,  
\begin{equation}
\tau(\{\rho_{a|x}\})\leq \Tr(\mathcal{R}_{\rho_B})
= d\tr(\rho_B)=d,\quad 
\bar{\tau}(\{M_{a|x}\})\leq \Tr(\bid\mathcal{R}_{\rho_B}\bid) =\frac{(d^2-1)\tr(\rho_B)}{d}=d-\frac{1}{d},
\end{equation}
 where $\mathcal{R}_{\rho_B}$
is defined as in \eref{eq:LRmultiplication} in the main text.
However, $\bid\mathcal{R}_{\rho_B}\bid$ is in general not equal to $\barcal{R}_{\rho_B}$  in \eref{eq:LRmultiplication}. 
\end{proof}
\begin{remark}
According to \pref{pro:taudif}, the upper bound in $\tau(\{M_{a|x}\})\leq d^2$ is saturated if and only if the upper bound in $\bar{\tau}(\{M_{a|x}\})\leq d^2-1$ is saturated. Similarly, the two upper bounds in $\tau(\{\rho_{a|x}\})\leq d$ and  $\bar{\tau}(\{\rho_{a|x}\})\leq d-1/d$ are simultaneously  saturated or not. 
\end{remark}

\begin{proof}[Proof of \pref{pro:MSAm}]
According to \lref{lem:SumBound} below and \eref{eq:Gbound} in the main text,
$\bar{\tau}(\{M_{a|x}\})\leq
\sum_x \Tr(\barcal{G}(\{M_{a|x}\}_a))\leq m(d-1)$. The first inequality is
saturated if and only if  the POVMs in the assemblage are mutually orthogonal,
and the second is saturated if and only if all the POVMs have rank 1. The
same proof applies to state assemblages.
\end{proof}

In the rest of this section we prove the two auxiliary lemmas needed in the
proofs of \psref{pro:MSA} and \ref{pro:MSAm}.
\begin{lemma}\label{lem:SLDbound}
Any  POVM $\{M_a\}$ satisfies  $\mathcal{G}(\{M_a\})\leq \mathbf{I}$ and
 $\barcal{G}(\{M_a\})\leq \bid$.
Any  ensemble $\{\rho_a\}$  for the state~$\rho_B$ satisfies  $\mathcal{G}(\{\rho_a\})\leq
\mathcal{R}_{\rho_B}$ and $\barcal{G}(\{\rho_a\})\leq \bid\mathcal{R}_{\rho_B}\bid$.
\end{lemma}
\begin{remark}
The bounds in the lemma are closely related to the quantum Cram\'er-Rao bound
based on the symmetric logarithmic derivative \cite{Zhu12the,Zhu15IC}.
\end{remark}
\begin{proof}
Let $C$ be an arbitrary Hermitian operator. Then \begin{align}
\dbra{C}\mathcal{G}(\{M_a\})\dket{C}=\sum_a\frac{\dinner{C}{M_a}\dinner{M_a}{C}}{\tr(M_a)}
=\sum_a\frac{\bigl[\tr\bigl(CM_a^{1/2} M_a^{1/2}\bigr)\bigr]^2}{\tr(M_a)}\leq
\sum_a \tr(C^2M_a)=\tr(C^2),
\end{align}
which implies that $\mathcal{G}(\{M_a\})\leq \mathbf{I}$. Conjugation by
$\bid$ yields $\barcal{G}(\{M_a\})\leq \bid$. By the same token, we have
\begin{align}
\dbra{C}\mathcal{G}(\{\rho_a\})\dket{C}\leq\tr(C^2\rho_B)=\dbra{C}\mathcal{R}_{\rho_B}\dket{C},
\end{align}
which implies that $\mathcal{G}(\{\rho_a\})\leq
\mathcal{R}_{\rho_B}$ and $\barcal{G}(\{\rho_a\})\leq \bid\mathcal{R}_{\rho_B}\bid$.
\end{proof}

\begin{lemma}\label{lem:SumBound}
Any assemblage $\{\rho_{a|x}\}$ satisfies
\begin{equation}
\bar{\tau}(\{\rho_{a|x}\})\leq \sum_x \Tr(\barcal{G}(\{\rho_{a|x}\}_a))=\sum_{a,x}\frac{\tr(\bar{\rho}_{a|x}^2)}{\tr(\rho_{a|x})}.
\end{equation}
The upper bound is saturated if and only if the ensembles in the assemblage
 are mutually orthogonal.
\end{lemma}

\Lref{lem:SumBound} is an immediate consequence of the following lemma.
\begin{lemma}
Suppose $A_j$ are positive operators in dimension $d$ and
\begin{equation}\label{eq:tfun}
t(\{A_j\}):=\min\{\tr(F)| F\geq A_j\;\forall j\}.
\end{equation}
Then $t(\{A_j\})\leq \sum _j \tr(A_j)$ and the upper bound is saturated if
and only if  the $A_j$ have mutually orthogonal support. In that case, the
operator that attains the minimum in \eref{eq:tfun} is unique and is equal
to $\sum_j A_j$.
\end{lemma}
\begin{proof}
Let $C=\sum_j A_j$; then $C\geq A_j$ for all $j$. So
\begin{equation}\label{aeq:tupperbound}
t(\{A_j\})\leq \tr(C)= \sum _j \tr(A_j).
\end{equation}

Let $E$ be an operator that attains the minimum in \eref{eq:tfun},  $P_j$
the projector onto the support of $A_j$, and $P^\bot$  the projector onto
the kernel of $\sum_j A_j$.  If the $A_j$ are mutually orthogonal, then the $P_j$
are mutually orthogonal and orthogonal to $P^\bot$. In addition, $P^\bot+\sum_j
P_j=1$ and $P_j E P_j\geq P_jA_jP_j=A_j$, so that
\begin{align}
t(\{A_j\})&=\tr(E)= \tr(P^\bot E P^\bot) +\sum_j \tr(P_j E P_j)\geq \sum_j \tr(P_j
E P_j)\geq\sum_j\tr(A_j),
\end{align}
which implies that $t(\{A_j\})= \sum _j \tr(A_j)$ given \eref{aeq:tupperbound}.
Furthermore,  saturating these inequalities implies that  $P^\bot E P^\bot=0$
and $P_j E P_j=A_j$ for all $j$. Given that $E-A_j$ is positive semidefinite
for all $j$, we conclude that $P^\bot E P_j=P_j E P^\bot =0$ for all
$j$ and $P_j E P_k=0$ for all $j\neq k$. Consequently, $E=\sum_j P_j E P_j=\sum_j
A_j$.

If two of the $A_j$, say $A_1$ and $A_2$, are not orthogonal, then there
exists a (nonzero) positive operator $B$ that satisfies $B\leq  A_1$ and $B\leq
A_2$. Let $C=\sum_j A_j-B$; then $C\geq A_j$ for all $j$. So
\begin{equation}
t(\{A_j\})\leq \tr(C)= \sum _j \tr(A_j)-\tr(B)< \sum _j \tr(A_j).
\end{equation}
\end{proof}

\section{Steering measures under transposition}
Here we show that $\tau(\{\rho_{a|x}\})$ and $\bar{\tau}(\{\rho_{a|x}\})$
are invariant under transposition of $\rho_{a|x}$.
\begin{lemma}\label{lem:tauTranspose}
\begin{equation}
\tau(\{\rho_{a|x}^\rmT\})=\tau(\{\rho_{a|x}\}), \quad \bar{\tau}(\{\rho_{a|x}^\rmT\})=\bar{\tau}(\{\rho_{a|x}\}).
\end{equation}
\end{lemma}
\begin{proof}
Let $\mathcal{F}$ be a superoperator that satisfies $\mathcal{F}\geq \mathcal{G}(\{\rho_{a|x}\})$
for all $x$ and that $\Tr(\mathcal{F})=\tau(\{\rho_{a|x}\})$. Decompose $\mathcal{F}$
as $\mathcal{F}=\sum_j r_j (\douter{R_j}{R_j})$ and let $\mathcal{F}_\rmT=\sum
r_j (\douter{R_j^\rmT}{R_j^\rmT})$. Then we have
\begin{equation}
\sum_j  r_j( \douter{R_j}{R_j})\geq \sum_a \frac{\douter{\rho_{a|x}}{\rho_{a|x}
}}{\tr(\rho_{a|x})}\quad  \forall x,
\end{equation}
which implies that
\begin{equation}
\mathcal{F}_\rmT\geq \sum_a \frac{\douter{\rho_{a|x}^\rmT}{\rho_{a|x}^\rmT
}}{\tr(\rho_{a|x})}=\mathcal{G}(\{\rho_{a|x}^\rmT\}_a) \quad \forall x.
\end{equation}
Therefore, $\tau(\{\rho_{a|x}^\rmT\})\leq \Tr(\mathcal{F}_\rmT)=\Tr(\mathcal{F})=\tau(\{\rho_{a|x}\})$.
By symmetry we also have $\tau(\{\rho_{a|x}\})\leq \tau(\{\rho_{a|x}^\rmT\})$.
It follows that $\tau(\{\rho_{a|x}^\rmT\})=\tau(\{\rho_{a|x}\})$. The equality
$\bar{\tau}(\{\rho_{a|x}^\rmT\})=\bar{\tau}(\{\rho_{a|x}\})$ follows from
the same reasoning.
\end{proof}

\section{Steerability of isotropic states}
Consider a family  isotropic states in dimension $d\times d$ parametrized by the parameter $\alpha$
with $0\leq \alpha\leq 1$,
\begin{equation}
\rho(\alpha)=(1-\alpha)\frac{1}{d^2}+\alpha |\Phi\rangle\langle \Phi|,
\end{equation}
where $|\Phi\rangle=(\sum_j |jj\rangle )/\sqrt{d}$.
Suppose Alice has the measurement assemblage $\{A_{a|x}\}$; then Bob has
the state assemblage $\{\rho_{a|x}\}$ with
\begin{equation}
\rho_{a|x}=\tr_A[(A_{a|x}\otimes 1)\rho(\alpha)]=\frac{1}{d}\Bigl[\alpha A_{a|x}^\rmT+\frac{1-\alpha}{d}\tr(A_{a|x})\Bigr],
\end{equation}
which may be seen as a coarse graining of the assemblage $\{A_{a|x}^\rmT/d\}$.
Calculation shows that
\begin{equation}
\bar{\tau}(\{\rho_{a|x}\})=\frac{\alpha^2}{d}\bar{\tau}(\{A_{a|x}^\rmT\})=\frac{\alpha^2}{d}\bar{\tau}(\{A_{a|x}\}),
\end{equation}
where the second equality follows from \lref{lem:tauTranspose}.
The isotropic state is steerable with respect to $\{A_{a|x}\}$ when $\alpha^2\bar{\tau}(\{A_{a|x}\})>
(d-1)$. If the measurement assemblage for Alice  is composed of $m$ MUB, then
$\bar{\tau}(\{A_{a|x}\})=m(d-1)$, so the isotropic state is steerable if
$m\alpha^2>1$. In the case of two qubits,  this condition turns out to be both
sufficient and necessary \cite{CavaJWR09,KogiSCA15}. Note that a two-qubit isotropic
state is equivalent to a Werner state under a local unitary transformation.

\end{document}